\documentclass[10pt,twocolumn]{article}

\usepackage{fullpage}
\usepackage[hyphens]{url}

\usepackage{amsmath,amssymb,bbm}
\usepackage{mathrsfs,enumerate,latexsym,graphicx}

\usepackage[switch,columnwise]{lineno}

\usepackage{authblk}
\usepackage{doi}

\usepackage[usenames,dvipsnames,svgnames,table]{xcolor}

\newcommand{\new}[1]{{\color{Black} #1}}
\newcommand{\newTwo}[1]{{\color{Black} #1}}
\newcommand{\newThree}[1]{{\color{Black} #1}}

\newtheorem{lemma}{Lemma}

\newtheorem{theorem}{Theorem}
\newtheorem{corollary}{Corollary}
\newtheorem{assumption}{Assumption}
\newtheorem{proof}{{\bf \sc Proof}}

\newcommand{\diff}{\mathrm{d}}
\DeclareMathOperator*{\E}{\mathbb{E}}
\newcommand{\prob}[1]{\mathbb{P}\left [ #1 \right ]}
\newcommand{\probSymbol}{\mathbb{P}}

\DeclareMathOperator*{\argmax}{arg\,max}

\newcommand{\degree}{m}

\newcommand{\threshold}{\tau}

\newcommand{\fracFunctional}{F}
\newcommand{\functionalState}{\texttt{F}}
\newcommand{\dysfunctionalState}{\texttt{D}}
\newcommand{\actualFracFunctional}{\mathcal{F}}
\newcommand{\numSuccesses}{S}
\newcommand{\chanceSuccess}[3]{P \left [ #1, #2, #3 \right ]}
\newcommand{\chanceSuccessSymbol}{P}

\newcommand{\slowTime}{T}

\newcommand{\utility}[5]{\mathcal{U} \left [ #1, #2, #3; #4, #5 \right ] }
\newcommand{\utilitySymbol}{\mathcal{U}}
\newcommand {\naturalNum} {\mathbb{N}_0}
\newcommand{\shift}{s}
\newcommand{\numNodes}{N}
\newcommand{\numNodesFunctional}{n}

\newcommand{\decay}{\epsilon}
\newcommand{\stickiness}{r}
\newcommand{\demanded}{u}

\begin{document}

\title{Contagious disruptions and complexity traps in economic development}

\author[a,b,*]{Charles D.\ Brummitt}
\author[c,d]{Kenan Huremovic}
\author[e]{Paolo Pin}
\author[b,f]{Matthew H.\ Bonds}
\author[e]{Fernando Vega-Redondo}
\affil[a]{\small Center for the Management of Systemic Risk, Columbia University, New York, NY, 10027, USA}
\affil[b]{Department of Global Health and Social Medicine, Harvard Medical School, Boston, MA 02115, USA}
\affil[c]{IMT School for Advanced Studies, Piazza S. Francesco, 19, 55100, Lucca, Italy}
\affil[d]{Aix-Marseille School of Economics, Aix-Marseille University, 5 Bd Maurice Bourdet, 13001 Marseille, France}
\affil[e]{Department of Decision Sciences, Innocenzo Gasparini Institute for Economic Research, and Bocconi Institute for Data Science and Analytics, Universit\`{a} Bocconi, Via Roentgen 1, Milano 20136, Italy}
\affil[f]{Earth System Science, Stanford University, Stanford, CA 94305, USA}
\affil[*]{Corresponding author: \href{mailto:charles_brummitt@hms.harvard.edu}{charles\_brummitt@hms.harvard.edu}}

\twocolumn[
\begin{@twocolumnfalse}
\maketitle
\begin{abstract}
Poor economies not only produce less; they typically produce things that involve fewer inputs and fewer intermediate steps. Yet the supply chains of poor countries face more frequent disruptions---delivery failures, faulty parts, delays, power outages, theft, government failures---that systematically thwart the production process. To understand \newThree{how these disruptions affect} 
economic development, we model an evolving input--output network in which disruptions spread contagiously among optimizing agents. The key finding is that a poverty trap can emerge: agents adapt to frequent disruptions by producing simpler, less valuable goods, yet disruptions persist. Growing out of poverty requires that agents invest in buffers to disruptions. These buffers rise and then fall as the economy produces more complex goods, a prediction consistent with global patterns of input inventories. Large jumps in economic complexity can backfire. \new{This result suggests why ``big push'' policies can fail, and it underscores the importance of reliability and of gradual increases in technological complexity.}
\vspace{.5cm}
\end{abstract}
\end{@twocolumnfalse}
]

Producing valuable goods and services is a complex, intricate process. One obtains inputs from a multitude of suppliers who must honour their contracts and deliver those inputs without them breaking, spoiling, or being stolen.  These inputs must be stored safely and manipulated in interdependent stages, using labour from workers who may fall ill or shirk their duties, together with complex equipment and vast infrastructure that may malfunction. These complex interdependencies underlie specialization and trade that are the foundation of economic growth and material progress\ \cite{schumpeter1934theory,romer1987growth}. 

Yet this progress, and the disruptions that thwart it, are unevenly distributed around the world. In low-income countries, disruptions can be frequent, long-lasting, and severe. They include power outages\ \cite{Allcott2014,EnterpriseSurvey}, worker absenteeism\ \cite{Chaudhury2006}, failed deliveries of products, water shortages, customs delays, damage from natural disasters, and epidemic disease (Fig.~\ref{fig:disruptions}). Poor countries also tend to produce simpler goods, especially primary resources like timber, mining, and subsistence agriculture\ \cite{Caselli2005,Hidalgo2007,Hidalgo2009}. 

\begin{figure}[htbp]
\begin{center}
\includegraphics[width=1.0\columnwidth]{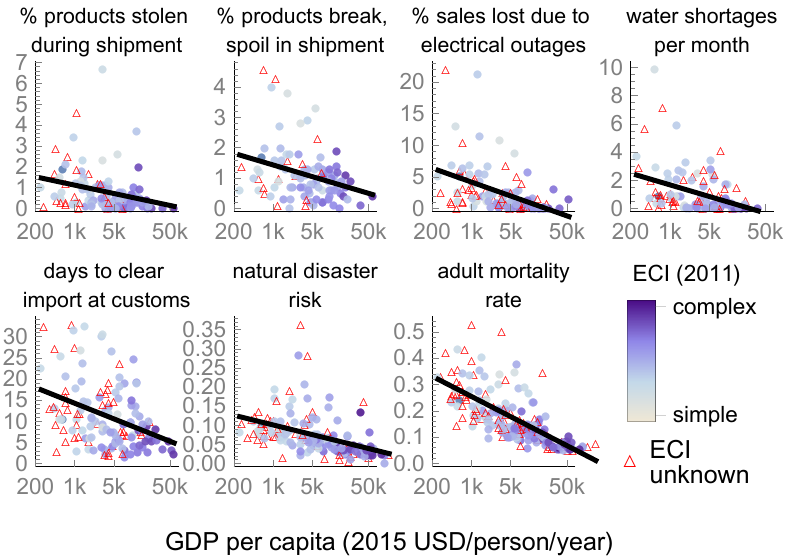}
\caption{Disruptions to the production process tend to be more frequent in poorer, less complex economies. 
The color of each dot indicates the country's Economic Complexity Index (ECI)\ \cite{ECI2011,AtlasEconomicComplexity};
\newThree{a red triangle is drawn is ECI is missing.} Black lines are least-squares fits, with per-capita incomes\ \cite{GDPperCapitaWorldBank} on a logarithmic scale. Natural disaster risk combines exposure and ability to cope\ \cite{NaturalDisasterRiskIndex}. Adult mortality rate is the chance that a $15$-year old dies before age $60$\ \cite{AdultMortalityRate}.
\newThree{Data in the first five plots are from\ \cite{EnterpriseSurvey}.}
}
\label{fig:disruptions}
\end{center}
\end{figure}

In middle- and high-income countries, by contrast, inputs tend to be more reliable, and goods produced tend to be more complex. But rich economies are not immune to disruptions: competition drives firms to build lean supply chains with buffers so small that disruptions can cascade around the globe, causing large aggregate losses\ \cite{Fortune2011,Punter2013}. 

Might the mechanisms causing globalised supply chains to become fragile also be preventing poor economies from becoming more complex and global? This question stretches the limits of our understanding of economic growth and complexity. Input--output linkages among firms---wherein one firm's output is another firm's input---are known to have large, nonlinear effects on economies. In theoretical models, these linkages propagate changes in productivity~\cite{Ciccone2002,Jones2011,Jones2011Misallocation,Acemoglu2012}, disruptions~\cite{Kremer1993}, and bankruptcies~\cite{Battiston2007,Weisbuch2007,Mizgier2012,Henriet2012,Levine2012,Contreras2014}. Empirical research on industrialised economies finds that supply-chain disruptions often lead to lower stock prices and sales growth\ \cite{Hendricks2003,Hendricks2005,Barrot2014,Wang2015}. 
\new{These disruptions, and the uncertainty that they entail, affect development: they cause firms to use less capital\ \cite{Kremer1993}, to misallocate inputs\ \cite{Jones2011,Jones2011Misallocation}, or to form shorter supply chains\ \cite{Levine2012}, generally hindering the economy to industrialise\ \cite{Ciccone2002} and limiting the effectiveness of policy\ \cite[Chapter 4]{Aoki2006_Reconstructing_book}.}
However, in these models, disruptions are treated as exogenous, and firms interact once in a static network. \new{These assumptions preclude the dynamic feedbacks that} can generate complex outcomes such as poverty traps and periodic cycles. 

\new{Modelling dynamic production networks is a challenging problem, involving heterogeneous input--output patterns and input elasticities\ \cite{Carvalho2014_Perspective}.}  
Recent models consider firms that endogenously form input--output linkages\ \cite{Duan:2011fy,Oberfield2012,Carvalho2014}; others consider firms deciding how to source their inputs in a risky supply chain with one\ \cite{Tomlin2006,Aydin2011} or more\ \cite{Bimpikis:2014vt,Ang:2014hd,Bakshi:2015wl} tiers. Missing is an understanding of how fast dynamics in economic networks, such as disruptions in supply chains, affect their long-run evolution and their growth in complexity.

We aim to fill this theoretical gap by introducing a simple model that captures complex dynamics of  disruptions spreading in an evolving input--output network. The main result is that poverty can emerge and reinforce itself: 
facing an unreliable environment of potential inputs, agents choose simple production processes that require few inputs, but disruptions remain frequent. Escaping this trap requires investing in buffers against disruption, such as arranging for extra suppliers or storing inventories of inputs. We find empirical support for the prediction that these buffers grow and then shrink as economies develop. When they shrink too much, disruptions can spike in number, as occurs in lean supply chains today. This mechanism also imperils developing economies: jumping abruptly to a more complex technology can backfire by causing greater dysfunction, suggesting that ``big push'' policies\ \cite{Murphy1989} may benefit from technological gradualism. We suggest that this alternative perspective---focused on contagion in evolving supply chains---may shed light on why poor economies are not catching up \newThree{and why some interventions fail.}


\section*{\newThree{Model of contagious disruption in an evolving input--output network}}
\label{sec:model}

We consider a large population of agents who represent entrepreneurs or firms producing goods and services that require inputs from other agents. The model framework is meant to correspond to a variety of situations that broadly represent the process of coordinating inputs and outputs for economic production: launching a business requires intermediate goods from suppliers; coordinating stakeholder meetings requires a quorum of attendees; repairing equipment requires parts supplied by others; and so on. 

\new{Our focus on fragility has antecedents in Kremer's ``O-Ring theory'' of economic development, in which a single mistake---like the malfunctioning O-Ring that triggered the explosion of the \emph{Challenger} Space Shuttle---can doom a sequence of interrelated tasks\ \cite{Kremer1993}. That study shows how fragility can lead to highly skilled workers matching with each other. Here, we focus on how people respond by using simpler technology or by investing in buffers against disruption so that some failures can be endured.}

\subsection*{Balls-and-urn model of production and contagious dysfunction}
\label{sec:balls_urn}
At each time $t$, all agents exist in one of two states. A fraction $\fracFunctional(t)$ of agents are \emph{functional}: they recently succeeded in producing and can provide inputs to others upon request. The remaining fraction $1 - \fracFunctional(t)$ are \emph{dysfunctional}: they recently failed to produce and cannot provide inputs to others. 

Agents become functional and dysfunctional as they succeed and fail, respectively, in producing goods or accomplishing tasks. Each agent attempts to produce a good requiring $\threshold$ many inputs. Attempts at producing a good occur randomly at a constant rate (as a Markov process). We do not track types of inputs nor economic sectors. This simplification allows us to abstract from which pairs of inputs are substitutes by using a simple threshold rule: an agent attempts to obtain inputs from $\degree$ agents in the population, and she succeeds in producing if and only if at least $\threshold$ of those $\degree$ many inputs are successfully produced and delivered to her (see Fig.\ \ref{fig:illustrate_model}). We call $\degree$ the number of \emph{attempted inputs}\newTwo{, and we think of it as} the \emph{in-degree} when viewing these interactions as an input--output network. 

\newThree{This threshold rule captures the idea that some inputs are critical: without them, production halts or fails. 
For example, the March 11, 2011 earthquake near Japan closed the Hitachi factory that produced most of the world's airflow sensors, a critical input for cars. As a result, automobile factories on the other side of the globe had to curb production or close~\cite{Punter2013}. In the developing world, drip irrigation has failed in Sub-Saharan Africa due to disruptions in water infrastructure and scarce knowledge for repair~\cite{Garb2014}; adulterated fertilizer sold in Ugandan markets yields negative average returns~\cite{Bold2015}; Internet-connected kiosks in India fell into disuse because of unreliable electricity and insufficient service from operators~\cite{InternetKiosksIndia}.}

\begin{figure}[htbp]
\begin{center}
\includegraphics[width=.6 \columnwidth]{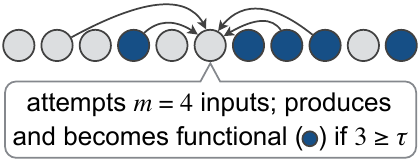}
\caption{Illustration of the model. Agents \new{(drawn as disks)} are people or firms who are either functional or dysfunctional at any moment in time. \emph{Functional} means that the agent has enough inputs needed to produce or to accomplish a task, and that other agents can rely on this agent for inputs. \new{Agents attempt} to produce a good (or to have a meeting with other people, etc.)\ using $\degree$ inputs drawn randomly from the population \new{(indicated by arrows)}, and \new{they succeed} if at least $\threshold$ of those suppliers are functional.}
\label{fig:illustrate_model}
\end{center}
\end{figure}

\newThree{The simplifications above}
allow us to describe an evolving input--output network, together with disruptions spreading on it, using a single differential equation for the expected fraction of functional agents $\fracFunctional(t)$: 
\begin{align}
\diff \fracFunctional / \diff t &:= 
\prob{\text{Binomial}\!\left (\degree, \fracFunctional(t) \right ) \geq \threshold} - \fracFunctional(t)
\label{eq:ODE_same_rate}
\end{align} 
for $t \geq 0$ and integers $\degree \geq 0$ and $\threshold > 0$. (If $\threshold = 0$, then there is little to model, so we let $\diff \fracFunctional / \diff t := 0$.) This framework is a balls-and-urn model~\cite{UrnBook,Woodbury1949} taken to an infinite-population, continuous-time limit, so that transition probabilities become deterministic rates of change in the mean-field master equation\ \eqref{eq:ODE_same_rate}\ \cite{MasterEquations}. \newThree{(Master equations are used to describe the dynamics of microeconomic actors in social science\ \cite{Helbing2010} and economics\ \cite{Aoki1998_NewApproaches_book,Aoki2004_modeling_book}.)}
We derive equation\ \eqref{eq:ODE_same_rate} in 
Supplementary Note 1
and explain it next. 

For simplicity, the input--output network is random and ``annealed'': in each production attempt, inputs are chosen uniformly at random with replacement from the population. This annealed network captures the idea that people do different tasks that require different inputs: an engineer fixes a machine on Monday and leads a meeting on Wednesday; an entrepreneur tries one business idea this year and another idea the next year, requiring different inputs for each step. \new{Later, we relax this assumption.} 

All agents use the same value of $(\degree, \threshold)$ (i.e., they play ``symmetric strategies''). Therefore, the chance of successful production is the probability $\probSymbol$ that a binomial random variable with parameters $\degree$ and $\fracFunctional(t)$ is $\geq \threshold$, where $\threshold$ is the number of critical inputs needed. This threshold rule resembles the essential inputs and critical subtasks in the ``O-Ring theory'' of Kremer\ \cite{Kremer1993}, but here people can have buffers against failures: the number of attempted inputs ($\degree$) can exceed the number of inputs needed ($\threshold$). This threshold rule also appears in models of social contagion and collective behavior\ \cite{Granovetter1978,Watts2002,Jackson2007},  but here we have an annealed network, bidirectional changes in state, and decision making, described later.

\new{Some disruptions may result not from other agents' dysfunction but from other causes, such as fires, insect outbreaks, weather, and so on. To capture these exogenous disruptions, we assume that all agents independently become dysfunctional for exogenous reasons at a small rate $\decay$ (according to a Poisson process). 
This assumption \newTwo{introduces} a $- \decay \fracFunctional(t)$ term to the master equation\ \eqref{eq:ODE_same_rate}:}
\begin{subequations}
\begin{align}
\diff F / \diff t
 &= \left [ 1 - \fracFunctional(t) \right ] \chanceSuccessSymbol - \fracFunctional(t) (1-\chanceSuccessSymbol + \decay)  \label{eq:ODE_definition_unsimplified} \\
 &= \chanceSuccessSymbol - \fracFunctional(t) (1 + \decay),\label{eq:ODE_definition} 
\end{align}
\label{eq:ODE_definition_subequations}
\end{subequations}
\!\!\! where $\chanceSuccessSymbol$ is the probability that an agent successfully produces, $\prob{\text{Binomial}\!\left (\degree, \fracFunctional(t) \right ) \geq \threshold}$. 
The first term in equation\ \eqref{eq:ODE_definition_unsimplified} is the rate $1-\fracFunctional(t)$ at which dysfunctional agents attempt to produce; each attempt succeeds with probability $\chanceSuccessSymbol$; if the attempt succeeds, then $\fracFunctional(t)$ rises; otherwise $\fracFunctional(t)$ stays the same, and vice versa for the second term. Equation \eqref{eq:ODE_definition_subequations} is derived in 
Supplementary Note 1,
and it recovers equation\ \eqref{eq:ODE_same_rate} \new{with $\epsilon := 0$}. 

The initial amount of dysfunction $1-\fracFunctional(0)$ is exogenous; after that, disruptions are entirely endogenous, spreading from supplier to customer. Driving this contagion is the assumption that an agent delivers an input upon request if and only if she successfully produced in her most recent attempt to produce. For example, a Ugandan farmer who discovers that her seeds were inauthentic\ \cite{Bold2015}; an Ethiopian farmer whose drip irrigation system fails because of upstream failures\ \cite{Garb2014}; or an automobile manufacturer who failed to produce due to missing parts\ \cite{Fortune2011,Punter2013} all may subsequently fail to deliver output promised to a customer.

\subsection*{Deciding on complexity $\threshold$ and on buffers against disruption $\degree - \threshold$}

The threshold $\threshold$ loosely captures the complexity of the good or service being produced: more complex goods require more inputs\ \cite{Ethier1982,Romer1990}. To capture the incentives to create high-value products, we present a simple, reduced-form model in which agents derive utility from successfully producing goods that require more inputs. 
\new{We assume that when an agent successfully produces, her induced utility grows with the complexity of production; for simplicity, we express this utility by $\threshold^\beta$ where $\beta \in (0,1)$. (This assumption, that complexity underlies rising productivity, is standard in economic models\ \cite{romer1987growth,Romer1990,Ethier1982}. For a derivation of $\threshold^\beta$, see
Supplementary Note 2.)
} 
We also assume that each attempted input costs $\alpha > 0$. This parameter $\alpha$ represents the marginal cost of finding suppliers, maintaining multiple suppliers for the same input\ \cite{ProvitiReport}, incentivising suppliers to have multiple manufacturing sites\ \cite{SupplyChainDisruptionsHBR2014}, or maintaining inventory of inputs\ \cite{Tang2006}\newThree{; for details, 
see Supplementary Note 3}.

For simplicity, we assume that each agent knows the current likelihood $\fracFunctional(t)$ that a uniformly-random supplier would successfully produce and deliver an input upon request. 
Based on that reliability $\fracFunctional(t)$, agents revise their strategy of how complex a product to produce ($\threshold \in \{0, 1, 2, \dots\}$) and how many inputs $\degree \in \{0, 1, 2, \dots\}$ to attempt to procure in order to produce that good. For instance, if suppliers are unreliable [i.e., $\fracFunctional(t)$ is small], then agents arrange for redundant inputs (i.e., $\degree - \threshold > 0$) provided that they can afford it. Agents must commit to a certain technology and production technique for a certain amount of time $\slowTime$, so we assume that every $\slowTime$ amount of time all agents simultaneously update their strategy to the ``best response'', the maximiser $(\degree^*, \threshold^*)$ of the utility function 
\begin{align}
\utility{\degree}{\threshold}{\fracFunctional(t)}{\alpha}{\beta}
:= \chanceSuccess{\degree}{\threshold}{\fracFunctional(t)} \threshold^\beta - \alpha \degree. \label{eq:utility}
\end{align}
Thus, agents' strategies at time $t$ are 
\begin{align}
\left (\degree^*, \threshold^* \right ) = \argmax_{\degree, \threshold \geq 0} \,\, \utility{\degree}{\threshold}{\fracFunctional(k T)}{\alpha}{\beta} 
\label{eq:best_response}
\end{align}
for $t \in [k \slowTime, (k+1)\slowTime)$ where $k \in \{0,1,2,\dots\}$. 
Together, equations\ \eqref{eq:ODE_definition_subequations}--\eqref{eq:best_response} and the initial $\fracFunctional(0)$ define the model.

We have abstracted from considerations about market equilibrium and price formation\new{, hence expressing all payoff magnitudes in terms of some fixed numeraire.} Only $\threshold$ goods are used in production, even if more than $\threshold$ of $\degree$ suppliers are functional, because unused inputs are assumed to be perfect substitutes for used ones. 
\newTwo{(In 
Supplementary Note 2.2,
we explain three alternative interpretations of the relationship between inputs and outputs in the production process.)}


\section*{Results}

Figure~\ref{fig:example_trajectory} illustrates the three phases of an economy in this model: trapped, emerging, and rich. To understand the figure, suppose that at time $t=0$ agents successfully produce and deliver an input only $50\%$ of the time [i.e., $\fracFunctional(0) = 50\%$]. Then, from equation\ \eqref{eq:utility}, agents choose the strategy $(\degree^*, \threshold^*) = (3,1)$, meaning that agents produce a good requiring $\threshold^* = 1$ input, but they arrange for $\degree^* - \threshold^* = 2$ extra suppliers because disruptions are common [$1 - \fracFunctional(0) = 50\%$]. Using this strategy in an economy with reliability $\fracFunctional(0) = 0.5$ causes disruptions to become less frequent ($\diff \fracFunctional / \diff t > 0$), indicated by the green curve marked ``$3,1$'' in Fig.~\ref{fig:example_trajectory}. 

\begin{figure}[htb]
\begin{center}
\includegraphics[width=\columnwidth]{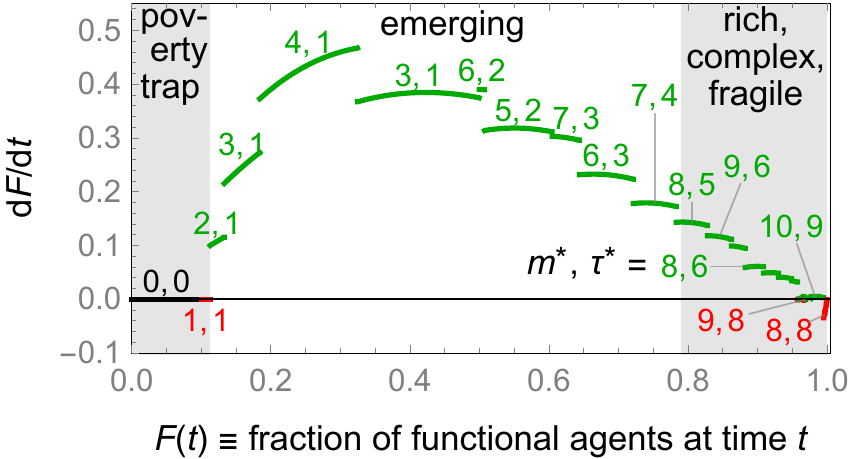}
\caption{
Representative phase portrait, showing the three phases of a model economy. Here, $(\alpha, \beta, \decay) := (0.1, 0.4, 0.001)$. 
The black, green, and red curves are the ODE\ \eqref{eq:ODE_definition_subequations} 
with labels indicating the best response $(\degree^*, \threshold^*)$ and colours denoting the sign of  $\diff \fracFunctional/ \diff t$. 
\newThree{Best responses are computed numerically using the method explained in 
Supplementary Note 4.}
}
\label{fig:example_trajectory}
\end{center}
\end{figure}

Figure\ \ref{fig:example_trajectory} corresponds to an economy in which agents best respond arbitrarily quickly based on the reliability $\fracFunctional(t)$ of their fellow agents; that is, the best-response timescale $\slowTime$ is arbitrarily close to $0$. This $\slowTime \to 0$ limit is more analytically tractable because $\diff \fracFunctional / \diff t$ changes discontinuously wherever the best response $(\degree^*, \threshold^*)$ changes as a function of $\fracFunctional(t)$. We relax this assumption later when we discuss \newTwo{fragility in rich economies}.

Next we explain the economy's three main phases and a pitfall in reaching the ``industrialised'' phase. 

\subsubsection*{Poverty trap with simple technology and frequent disruptions}
In an economy with frequent disruptions [$\fracFunctional(t)$ near zero], agents choose to withdraw from the economy by not relying on any inputs from others ($\degree^* = \threshold^* = 0$). This strategy resembles subsistence agriculture, hunting, and pastoralism. Such an economy is in steady state: $\diff \fracFunctional / \diff t$ \newTwo{ is equal to} $0$ from equation\ \eqref{eq:ODE_definition_subequations}, and no agent wants to deviate from the strategy $(0,0)$. 

\new{This steady state also has a basin of attraction (marked \lq\lq poverty trap\rq\rq in Fig.~\ref{fig:example_trajectory}), provided that there are some exogenous sources of disruption ($\decay > 0$). (\new{Specifically, for $\fracFunctional$ just above $\alpha$, the best response is $(\degree^*, \threshold^*) = (1, 1)$. That strategy means attempting a task that requires $\threshold^* = 1$ input and requesting that one input from $\degree^* = 1$ other agent. This strategy has no redundant inputs. It succeeds in producing with probability $\chanceSuccess{1}{1}{\fracFunctional} = \fracFunctional$, so,} \new{from equation \eqref{eq:ODE_definition_subequations}, $\diff \fracFunctional / \diff t = \fracFunctional -\fracFunctional - \decay < 0$.}})

\subsubsection*{Emerging economies' buffers to disruption rise and then fall}
If an economy is sufficiently reliable then it begins to develop. For instance, in Fig.~\ref{fig:example_trajectory} if $\fracFunctional(t) > 2^{-(\beta +1)} \left(1 - \sqrt{1 - \alpha  2^{\beta +2}}\right) \approx 11\%$  then agents 
choose to produce goods that require some inputs ($\threshold^*>0$). Provided that $\fracFunctional(t)$ is not too close to one (a case described later), the agents also arrange for some extra inputs ($\degree^* > \threshold^*$) in anticipation that some inputs will not be functional. This strategy results in the economy becoming more functional over time ($\diff \fracFunctional / \diff t > 0$) and producing ever more complex goods [$\threshold^*$ rises with $\fracFunctional(t)$]. 

As this economy develops, two features rise and then fall over time: the speed of development $\diff \fracFunctional / \diff t$ and the buffer against disruptions $\degree^* - \threshold^*$. (Later we examine this inverted-U pattern empirically.) This inverted-U reflects the \newTwo{following ideas.} Firms in a very unreliable economy need costly buffers against disruption to produce even simple, low-value goods (\newTwo{such as goods} with complexity $\threshold^* = 1$). The economy barely manages to produce \newTwo{such} simple goods using the small amounts of redundancy afforded by the low earnings (for example, with redundancy $\degree^*- \threshold^* = 2$). When the economy is more reliable [higher $\fracFunctional(t)$], more complex tasks become feasible with large buffers against disruption, such as complexity $\threshold^* = 3$ with buffer $\degree^* - \threshold^* = 4$. Finally, as the economy becomes maximally reliable [$\fracFunctional(t)$ approaches one], agents economise on their costly buffer against disruptions ($\degree^* - \threshold^*$), which leads to new vulnerabilities. 

\subsubsection*{Rich yet \newTwo{fragile}} 

When the economy becomes very reliable [large $\fracFunctional(t)$], agents produce very complex goods requiring many inputs (large $\threshold^*$). Yet this high reliability also induces agents to economise on their buffers to disruptions. In fact, when $\fracFunctional(t)$ is close to $1$, they eliminate their buffer ($\degree^* = \threshold^*$). Then disruptions spread like a virus to which no one is immune: $\fracFunctional(t)$ falls, indicated by the red curves in the bottom-right corner of Fig.~\ref{fig:example_trajectory}, where $\diff \fracFunctional / \diff t < 0$. Falling $\fracFunctional(t)$ means that more and more agents are unable to produce, and the drop in output resembles a recession. Such downturns occur generically in rich, highly functional economies: 
Theorem 1 in Supplementary Note 5 
\newTwo{shows that the state $\fracFunctional = 1$ (a completely functional economy) is unstable to perturbations. 
\newThree{In our model, a rich, highly functional economy is ``fragile'' in the sense that there are large values of $\fracFunctional(t)$ close to $1$ for which $\diff F / \diff t < 0$.} 
This ``rich yet fragile'' phenomenon accords with recent events: firms face pressure to build leaner supply chains, to invest in smaller buffers against disruptions, and to produce ever more complex goods, resulting in occasional cascading disruptions\ \cite{Fortune2011,Punter2013}.
}

\newTwo{
What happens after $\fracFunctional(t)$ begins to fall depends on how quickly agents best respond and whether the best response is discrete. 
}
If agents commit to a strategy for a positive amount of time $\slowTime > 0$, then the economy's reliability $\fracFunctional(t)$ falls until either (i) the economy enters the poverty trap (which occurs only for very large $\slowTime$) or (ii) agents best respond in a way that causes $\fracFunctional(t)$ to begin to rise. 
$\fracFunctional(t)$ can rise because agents produce simpler, lower-value goods (smaller $\threshold$) or because they increase the buffer against disruption (larger $\degree$). 
\newTwo{If $\degree$ and $\threshold$ are discrete (as considered here) and $\slowTime > 0$, then the economy can cycle:} After $\fracFunctional(t)$ rises for a while, agents best respond again, and because their economy is quite reliable they choose to produce very complex goods or to decrease their buffer against disruption, and the process \newTwo{can repeat}. \newTwo{If the decision variables $\degree$ and $\threshold$ were made continuous, or if $\slowTime \to 0$, then economy may settle onto a value of $\fracFunctional$ smaller than one.}

This fragility of rich economies complements theories of ``aggregate fluctuations''\ \cite{Bak1993,Gabaix2011,Acemoglu2012,Burlon2012,Baqaee2015,Carvalho2014_Perspective}. Those theories show how exogenous shocks to firms can result in large changes in the total production in the economy. 
One reason is heterogeneity: some firms and sectors are much larger\ \cite{Gabaix2011} or more connected\ \cite{Acemoglu2012,Burlon2012} than others, so a small shock to these important firms can have large consequences. 
\newThree{The models in\ \cite{Gabaix2011,Acemoglu2012,Burlon2012,Baqaee2015} are static and timeless, whereas our model is inherently dynamic, with most ``shocks'' caused by the endogenous failure of other firms.} 
Another reason \newThree{for aggregate fluctuations} is that firms' inventories self-organize to a critical point\ \cite{Bak1993}. \newThree{In the model in\ \cite{Bak1993}, firms request inputs from suppliers, and these requests spread through a fixed network. Here, firms also request inputs from suppliers, but disruptions (i.e., failures to produce due to insufficient functional inputs) spread contagiously in a network that changes over time.}

\subsubsection*{Overshooting complexity can backfire} 

The core mechanism that causes downturns in the rich economy also makes it difficult for emerging economies to become rich and complex themselves. Specifically, if an economy tries to ``prematurely jump'' to a more complex technology, then it can slide backward and become more dysfunctional. To make this idea precise, suppose that agents do not use the best response $(\degree^*, \threshold^*)$ but instead attempt a more complex strategy that requires $\shift$ more inputs: $(\degree^* + \shift, \threshold^* + \shift)$. The buffer against disruptions remains the same; it is still $\degree^* - \threshold^*$. What is different is that agents try to produce goods that require more inputs, or they try to produce the same good as before but using technology that depends on more inputs, such as drip irrigation instead of traditional irrigation~\cite{Garb2014}. 

\begin{figure}[hbt]
\begin{center}
\includegraphics[width=\columnwidth]{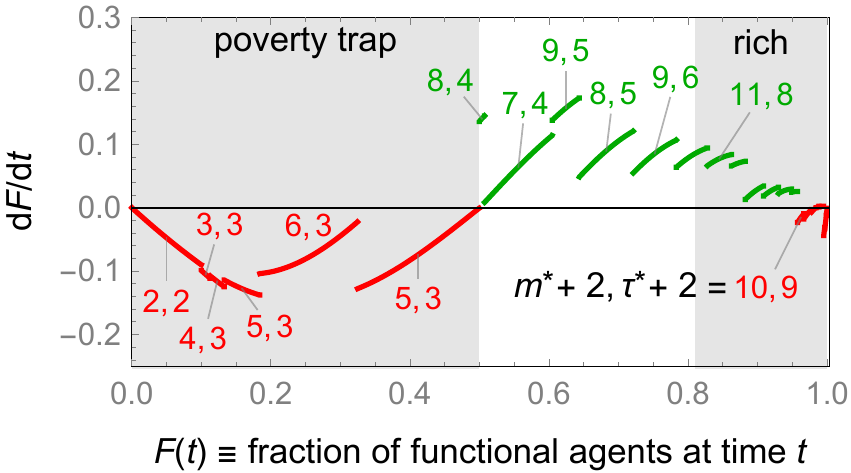}
\caption{
Jumping to a more complex technology can backfire by causing dysfunction to rise, especially for emerging economies. Here, agents use a strategy that requires two more inputs than they would have chosen given the reliability $\fracFunctional(t)$ of their potential inputs: they use the strategy $(\degree^* + 2, \threshold^* + 2)$, where $(\degree^*, \threshold^*)$ is the best response. The parameters are the same as in Fig.~\ref{fig:example_trajectory}.}
\label{fig:overshoot}
\end{center}
\end{figure}

Figure~\ref{fig:overshoot} shows that this strategy $(\degree^* + 2, \threshold^* + 2)$ often results in dysfunction rising over time ($\diff \fracFunctional / \diff t < 0$), indicated by the red curves. In these intervals with $\diff \fracFunctional / \diff t < 0$, agents are ``overshooting'' in complexity: they are attempting a production process more complex than what the surrounding system can support. 
This overshooting echoes failures to adopt complex technologies in developing countries because the technologies depend on myriad inputs prone to disruption, such as drip irrigation systems~\cite{Garb2014} and Internet kiosks~\cite{InternetKiosksIndia}. 

Emerging economies are especially vulnerable to overshooting in complexity: notice in Fig.\ \ref{fig:overshoot} that $\diff \fracFunctional / \diff t < 0$ for many intermediate values of $\fracFunctional(t)$. As a result, the poverty trap in Fig.\ \ref{fig:overshoot} is dramatically larger than when agents use the best response (compare with Fig.~\ref{fig:example_trajectory}). For example, an economy with $\fracFunctional(t)$ near $50\%$ can fall into the poverty trap if it overshoots in complexity for a sufficiently long amount of time. By contrast, a rich economy can typically accommodate a jump in complexity without causing dysfunction to rise: in Fig.~\ref{fig:overshoot} there are many large values of $\fracFunctional(t)$ with $\diff \fracFunctional / \diff t > 0$. 

Comparing Figs.\ \ref{fig:example_trajectory} and\ \ref{fig:overshoot}, we see the benefit of gradual growth in technological complexity. This prescription is at odds with the classic idea of a ``big push'' of simultaneously industrialising many sectors of an economy\ \cite{Murphy1989}: a big push overcomes coordination problems, but it can add fragility by introducing complex technologies that depend on unreliable inputs. This prescription for slow, gradual reform mirrors the suggestions given by a model of trust and social capital\ \cite{Francois2005}.

\subsubsection*{Phase diagram}
To demonstrate that the phenomena in Fig.~\ref{fig:example_trajectory} are rather generic, Fig.~\ref{fig:phase_diagram} shows the sign of $\diff \fracFunctional / \diff t$ and the best response $(\degree^*, \threshold^*)$ for many values of the parameter $\alpha$, the marginal cost of each attempted input. A poverty trap occurs for $\fracFunctional(t) \in [0, \alpha]$; the boundary $\fracFunctional(t) = \alpha$ is the indifference curve between $(\degree, \threshold) = (1,1)$ and $(0,0)$. If the cost to arrange for an input is too high ($\alpha >1/4$ in Fig.~\ref{fig:phase_diagram}), then the only long-run outcome is poverty. Otherwise, there exists a good outcome in the long run in which the economy is complex and highly functional, and buffers against disruptions $\degree^* - \threshold^*$ tend to rise and then fall as the economy approaches this rich state. 

\begin{figure}[htb]
\begin{center}
\includegraphics[width=.85\columnwidth]{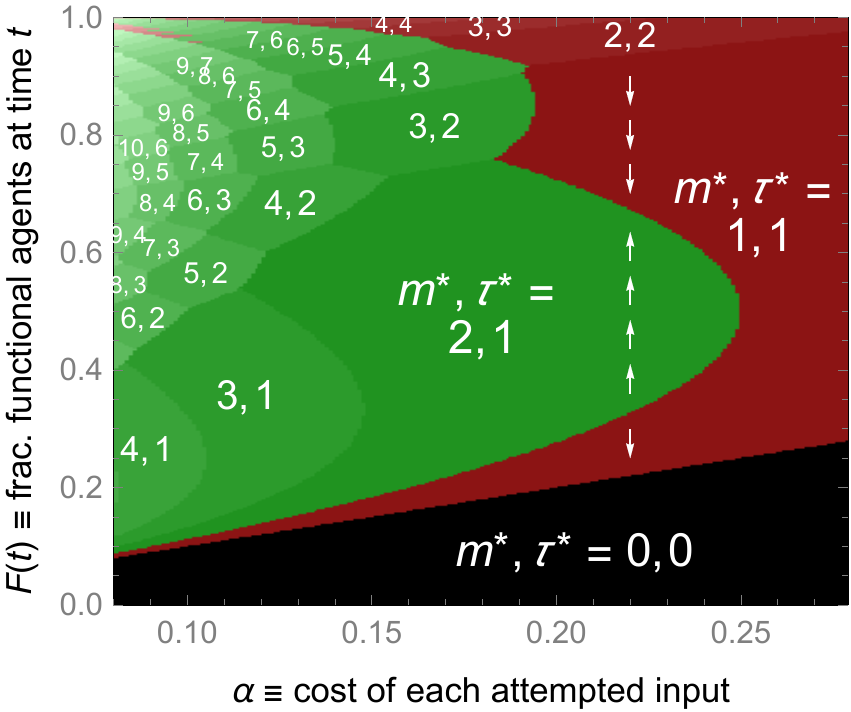}
\caption{Phase diagram of the sign of $\diff \fracFunctional / \diff t$ ($\diff \fracFunctional / \diff t > 0$ is green; $< 0 $ is red; $= 0$ is black) and the best response $(\degree^*, \threshold^*)$ (white text) as a function of $\alpha$ and $\fracFunctional(t)$. The parameter $\alpha$ is assumed to change slowly, if at all. Vertical white arrows illustrate the dynamics $\diff \fracFunctional/\diff t$. Here, as in Fig.~\ref{fig:example_trajectory}, $\beta := 0.4$, and \new{$\decay = 0.001$}, so the strategy $(\degree^*, \threshold^*) = (1,1)$ results in $\diff \fracFunctional/\diff t < 0$. Poor economies get stuck in the region labelled ``0,\,0''. Rich, complex economies settle upon limit cycles near where green lies below red at the top of the diagram.
}
\label{fig:phase_diagram}
\end{center}
\end{figure}

However, there are pitfalls in reaching this rich state. One pitfall is the ``overshooting'' described above. Another is to decrease the cost $\alpha$ of each attempted input. The marginal cost $\alpha$ is exogenous, but it could change if, for example, communications technology makes it easier to arrange alternative suppliers. Decreasing $\alpha$ can trigger an escape from the poverty trap if it puts the economy in the green region in Fig.~\ref{fig:phase_diagram}. But it can also make the economy more dysfunctional: if $\alpha$ is decreased into the red region, where agents choose $\degree^* = \threshold^* = 1$, then dysfunction rises \new{(provided that exogenous failures occur, i.e., $\decay > 0$)}. The intuition is that decreasing $\alpha$ incentivises people to attempt more complex production that uses more inputs (higher $\threshold^*$), which can lead to more failure than success, resulting in more frequent dysfunction in the new steady state [lower $\fracFunctional(t)$]. If policymakers sense this feedback, then they may avoid actions that decrease $\alpha$, keeping the economy stuck in the trap. 

\new{
\subsubsection*{Countervailing effects of keeping functional suppliers and of choosing popular suppliers}

The model presented above is simplified by the assumption that agents choose new suppliers uniformly at random every time they try to produce. At the other extreme, many models of economic cascades assume a rigid input--output network\ \cite{Acemoglu2012,Battiston2007,Weisbuch2007,Henriet2012,Levine2012,Contreras2014}. 

To explore a more realistic middle ground between these extremes, 
in Supplementary Note 6
we modify the choice of suppliers in two ways: agents tend to keep functional suppliers, and they bias their search toward suppliers who already have many customers (i.e., preferential attachment). 
These changes do not affect the qualitative results insights of the model, but they do have two interesting effects that we illustrate using numerical simulations in 
Supplementary Figure 1.

One effect is that the economy is less likely to fall into the trap. It is straightforward that a tendency to retain functional suppliers helps $\fracFunctional(t)$ grow. More interestingly, a tendency to choose popular suppliers also helps $\fracFunctional(t)$ grow: because functional agents tend to accumulate customers, having many customers is correlated with being functional. 

However, these two tendencies can generate fragility. Once the economy is complex and highly functional, it can rely on very few agents who supply almost everyone. When those ``supplier-hubs'' become dysfunctional (because they rely on dysfunctional suppliers or because they were suffered a rate-$\decay$ exogenous failure), then the brittle economy can undergo a severe downturn and cascading disruptions. In summary, what makes an economy more likely to emerge from the pull of poverty is precisely what makes the economy fragile upon becoming complex.
}

\subsection*{Empirical support}
\subsubsection*{Input inventories rise and then fall as economies become more complex}

There is scant data---especially in developing countries---on supply-chain disruptions and on responses to them. 
Relevant data from the World Bank's Enterprise Survey include the number of days of inventory that firms keep of their ``main input'' (i.e., highest-value input)~\cite{EnterpriseSurvey}. Stockpiling inputs is one costly way to mitigate the risk of disruptions in one's supply chain\ \cite{Tomlin2006}, so it loosely corresponds to our model's buffer against disruption $\degree^* - \threshold^*$. Macroeconomic research has focused on inventories of finished goods, but inventories of inputs have drawn increasing attention\ \cite{Humphreys2001,Iacoviello2011}, and some models of input inventories also consider intermediate goods and supply chains\ \cite{Khan2007,Wen2011}. 

\begin{figure}[htb]
\begin{center}
\includegraphics[width=.9\columnwidth]{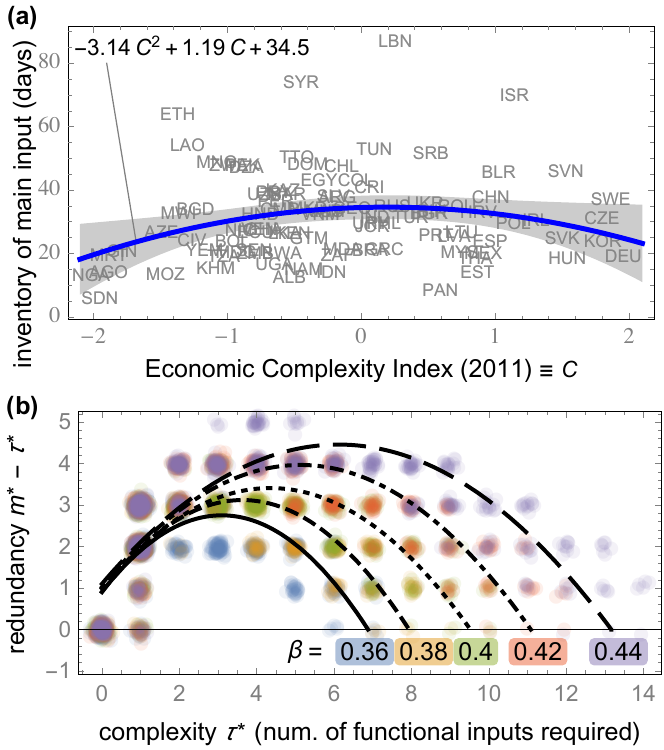}
\caption{Qualitative match between (a) empirical data on input inventories and (b) the model's prediction that buffers to supply-chain disruptions rise and then fall as economies develop. (a) Input inventories of firms, averaged at the country level~\cite{EnterpriseSurvey}, have an inverted-U relationship with the complexity of the economy~\cite{ECI2011,AtlasEconomicComplexity} ($p$-value $0.022$ for the $C^2$ coefficient; $R^2 = 0.063$; $N = 95$ countries; $95\%$ mean prediction band shown in gray). (b) Redundancy versus complexity for $\alpha := 0.1$, $\decay := 0.001$, and five values of $\beta$, each a different colour. The curves show least-squares fits to $\delta_0 + \delta_1 \threshold^* + \delta_2 \threshold^{*2}$. The data is dispersed by $\mathcal{N}(\mathbf{0},0.008 \times \mathbf{1})$ to indicate density.}
\label{fig:inventory}
\end{center}
\end{figure}

Recall from Fig.~\ref{fig:example_trajectory} that our model predicts that buffers against disruptions to inputs ($\degree^* - \threshold^*$) tend to rise and then fall as economies develop. 
To test this qualitatively, we plot  in Figure~\ref{fig:inventory}(a) 
the input inventory of firms averaged at the country level, for $95$ countries for which we have an estimate of the complexity of the economy~\cite{ECI2011,AtlasEconomicComplexity}. This Economic Complexity Index is calculated from the bipartite network of countries and of the products that they export~\cite{Hidalgo2009}. 
We find that input inventory has a statistically significant inverted-U relationship with the complexity $C$ of an economy's production. [A least-squares fit of inventory with $\gamma_0 + \gamma_1 C + \gamma_2 C^2$ has $\hat \gamma_2 = -3.14$, with $p\text{-value} = 0.02$; see Fig.~\ref{fig:inventory}(a).] 
This relationship qualitatively matches the inverted-U exhibited by the model [Fig.~\ref{fig:inventory}(b)]. 

\new{
\subsubsection*{Case study: Drip irrigation}

The reported uneven success of {drip irrigation} in different countries\ \cite{Garb2014} illustrates the assumptions and messages of this model. 
Drip irrigation applies water directly to roots at a small, consistent rate, which increases efficiency and 
transforms land from arid to arable. The technology is delicate and complex because it depends on a broad system of inputs: it requires high-quality water to be delivered at the right pressure, in pipes and tubes that match the local soil, crop and weather. Its equipment needs expert advice and repair. 
There is little buffer against malfunction because crops fail quickly in dry soil if the water's flow is interrupted. In our model, drip irrigation resembles a high threshold (high $\threshold$) technology (compared to rainwater), for which there is little buffer against disruptions ($\degree - \threshold$). 

Drip irrigation has been ``spectacularly successful'' in Israel, but ``the very same hardware often turned out to be completely useless in the sub-Saharan African context''\ \cite[p.\ 14]{Garb2014}. 
Farmers in Israel enjoy an ``extensive infrastructural network so pervasive and successful as to be nigh invisible''\ \cite[p.\ 16]{Garb2014}. 
Their counterparts in Ethiopia and Zambia had equipment from the same company, but they faced problems in the surrounding socio-technical system (of expertise, supply of water, and so on). Many farmers in Sub-Saharan Africa are ``gradually convert[ing] [their farms] back to furrow irrigation as each block of the buried drip irrigation fails''\ \cite[p.\ 19]{Garb2014}. 
Expressed in the language of our model, dysfunctional inputs (low $\fracFunctional$) can result in crops failing and farmers choosing simpler, less productive technology (lower $\threshold$). 
Drip irrigation has succeeded in poor regions (such as in India) not because the technology was simplified but because of sufficient support from the surrounding socio-technical system\ \cite[p.\ 22]{Garb2014}.
}


\section*{Discussion} 

Poverty traps have long been used to explain disparities of incomes across countries and to justify a ``big push'', a coordinated investment in many sectors to unleash growth\ \cite{Murphy1989,Sachs2004,Azariadis2005,Barrett2016}.  
Yet many big pushes have failed\ \cite{easterly2006reliving}\new{, and understanding why is paramount. }

\new{
Our model views industrialisation as mutually-reinforcing supply chains, broadly defined, that become more complex over time. 
Disruptions in these supply chains can spread contagiously. } This systemic fragility can cause complex technologies to fail. 
Even if all firms coordinate their industrialization (as suggested by big push theories\ \cite{Murphy1989,Azariadis2005}), if the firms jump too far in technological complexity without sufficient buffers against disruptions, then the economy can slide backward, becoming poorer and less reliable. As in other complex systems\ \cite{Gershenson2015}, going slower may result in collectively going faster.

\newThree{
This work sits at the intersection of competing theories of economic development. According to research on poverty traps\ \cite{Murphy1989,Sachs2004,Azariadis2005,Barrett2016}, positive feedback loops keep populations stuck in poverty, and escaping these traps requires substantial investments. According to research on institutions, differences in the rules of the game (such as property rights and rule of law) explain why economies have diverged\ \cite{north1990institutions,Acemoglu2005}. Our model provides a bridge between these views. Because some inputs for production are from government actors, our model is consistent with the institutional view of development: better institutions may imply fewer disruptions, less uncertainty, and hence greater appetite for complexity. Consistent with the poverty trap literature, our model has multiple equilibria. However, whereas poverty traps typically suggest making a large investment, we find reason for caution: without a focus on reliability of the surrounding system, big changes fail. Interventions to many parts of a system may be needed.} Unreliability affects economic performance in a multifaceted way, involving risk\ \cite{Morduch1994}, network contagion, technology adoption\ \cite{Garb2014}, and psychology\ \cite{PsychologyOfPoverty}. Understanding their interplay can elucidate the causes of persistent poverty.

\begin{description}
 \item [Acknowledgments] C.D.B. and M.H.B. acknowledge funding from the James S. McDonnell Foundation for the Postdoctoral Award and the Scholar Award (respectively) in Complex Systems. P.P. and F.V.-R. acknowledge funding from the Italian Ministry of Education Progetti di Rilevante Interesse Nazionale (PRIN) grant 2015592CTH. 
 The funders had no role in the conceptualization, design, data collection, analysis, decision to publish, nor the preparation of the manuscript.
 \item [Author Contributions] All authors designed research, performed research, and wrote the paper; C.B. and K.H. analyzed data.
 \item [Competing Interests] The authors declare that they have no
competing interests.
\end{description}

\section*{\newThree{Methods}}
\newThree{
In Supplementary Note 4, 
we compute the finite set of strategies that could be a best response for a given $\alpha$, $\beta$, and $\fracFunctional(t)$. This derivation enables the computations used to make Figures\ \ref{fig:example_trajectory},\ \ref{fig:overshoot},\ \ref{fig:phase_diagram}, and\ \ref{fig:inventory}(b).
}

\subsection*{Data availability}
The empirical data used in Figure\ \ref{fig:disruptions}\ and Figure\ \ref{fig:inventory}(a) are \newThree{available from the original sources\ \cite{ECI2011,AtlasEconomicComplexity,GDPperCapitaWorldBank,NaturalDisasterRiskIndex,AdultMortalityRate,EnterpriseSurvey} and \href{https://doi.org/10.5281/zenodo.823260}{at the GitHub repository associated with this paper} (\href{https://doi.org/10.5281/zenodo.823260}{DOI \texttt{10.5281/zenodo.823260}})}. 

\subsection*{Code availability}
The code used to produce the results in this paper are \newThree{\href{https://doi.org/10.5281/zenodo.823260}{available at the GitHub repository associated with this paper} (\href{https://doi.org/10.5281/zenodo.823260}{DOI \texttt{10.5281/zenodo.823260}})}. 
\new{Figures\ \ref{fig:disruptions}--\ref{fig:inventory} are created in the Wolfram Language (version 11) and can be read for free using the CDF Player and run for free using the Wolfram Cloud. \newThree{Supplementary} Figure 1 was created using Python (3.5.2) and NumPy (1.11.3). 
}

\clearpage 
\newpage

\onecolumn
\begin{center}
{\huge Supplementary Information}
\end{center}

\setcounter{figure}{0}
\setcounter{section}{0}
\makeatletter 
\renewcommand{\thefigure}{SI-\@arabic\c@figure}
\renewcommand{\theequation}{SI-\@arabic\c@equation}
\renewcommand{\thetable}{SI-\@arabic\c@table}
\renewcommand{\thesection}{SI-\@arabic\c@section}
\makeatother

\section{Derivation of the ODE\ \eqref{eq:ODE_definition_subequations}}
\label{sec:derive_ODE}

Suppose that we have a large number of agents $\numNodes \in \naturalNum$, and let $\numNodesFunctional(t)$ be the number of agents who are functional at time $t$, so that 
\begin{align*}
\actualFracFunctional(t) \equiv \frac{\numNodesFunctional(t)}{\numNodes}
\end{align*} 
is the fraction of agents who are functional at time $t$. $\actualFracFunctional(t)$ is a stochastic process. Here we derive a mean-field approximation for the master equation of $\actualFracFunctional(t)$, which is an ordinary differential equation (ODE) for the expected fraction of functional agents at time $t$, $\fracFunctional(t) \equiv \E \actualFracFunctional(t)$. 

\subsection{Transition probabilities for a single agent}
Attempts at producing are events that occur randomly according to a Markov process. All agents attempt to produce at rate $1$. 

Agents also become dysfunctional at rate $\decay$ for ``exogenous'' reasons that are not explicitly modeled. (We have in mind that these exogenous failures are isolated incidents such as fires, natural disasters, insect infestations, and so on.) These exogenous failures occur independently to each agent according to a Poisson process with rate $\decay$. 

First we focus attention on a single agent. Let $\diff t$ be a small, positive amount of time. 
Let $\omega(\functionalState \to \dysfunctionalState) \diff t$ be the chance that an agent that is functional at time $t$ is dysfunctional at time $t+\diff t$. Similarly, let $\omega(\dysfunctionalState \to \functionalState) \diff t$ be the probability that a dysfunctional agent at time $t$ is functional at time $t + \diff t$. 

There are many ways in which a functional agent at time $t$ could become dysfunctional at time $t + \diff t$. In that short amount of time, this agent could do the following transitions: 
\begin{enumerate}
\item $\functionalState \to \dysfunctionalState$: the agent tries to produce once and fails. This event occurs with probability $\diff t (1-P) + \mathcal{O}(\diff t^2)$.
\item $\functionalState \to \dysfunctionalState$: the agent becomes dysfunctional for an exogenous reason. This event occurs with probability $\decay \diff t + \mathcal{O}(\diff t^2)$.
\item $\functionalState \to \functionalState \to \dysfunctionalState$: the agent tries to produce twice, succeeding on the first try and failing on the second try, and does not suffer any exogenous failures. This event occurs with probability $\mathcal{O}(\diff t^2)$.
\item $\functionalState \to \dysfunctionalState \to \functionalState \to \dysfunctionalState$: the agent tries to produce three times, succeeding on only the second of those tries, and the agent does not suffer any exogenous failures. This event occurs with probability $\mathcal{O}(\diff t^2)$.
\item And so on.
\end{enumerate}
The chance of the first event is the chance that this agent is chosen to attempt to produce only once in the amount of time $\diff t$, which occurs with probability $\diff t e^{-\diff t} = \diff t + \mathcal{O}(\diff t^2)$, times the chance that the agent fails to produce in that attempt, which is $1-\chanceSuccessSymbol$, times the chance that no exogenous failure hit the agent, which is $1 - \decay \diff t + \mathcal{O}(\diff t^2)$. In sum, this event occurs with probability $\diff t (1-P) + \mathcal{O}(\diff t^2)$.

The chance of the second event above is the chance that the agent is not chosen to attempt to produce in that time interval $\diff t$, which is $1 - \diff t$, times the chance of failing for exogenous reasons in an amount of time $\diff t$, which is $\decay \diff t e^{-\decay} = \decay \diff t + \mathcal{O}(\diff t^2)$. 

The probabilities of all the other events are $\mathcal{O}(\diff t^2)$ because they require getting chosen at least twice to attempt to produce. Therefore, to first approximation in this small, positive amount of time $\diff t$, the chance that a functional agent at time $t$ is dysfunctional at time $t+\diff t$ is
\begin{subequations}
\begin{align}
\omega(\functionalState \to \dysfunctionalState) \diff t &= \diff t \times (1-\chanceSuccessSymbol + \decay) + \mathcal{O}(\diff t^2).
\end{align}

Similarly, the chance $\omega(\dysfunctionalState \to \functionalState) \diff t$ that a dysfunctional agent at time $t$ is functional at time $t + \diff t$ is the chance $1 \times \diff t$ that this agent is chosen to produce once in the time interval $(t, t+\diff t)$, times the probability $\chanceSuccessSymbol$ of successfully producing in that one attempt, times the chance $1 - \decay \diff t$ of not being hit by an exogenous failure in time $\diff t$, plus higher-order terms $\mathcal{O}(\diff t^2)$. Thus,
\begin{align}
\omega(\dysfunctionalState \to \functionalState) \diff t &= \diff t \times \chanceSuccessSymbol + \mathcal{O}(\diff t^2).
\end{align} 
\label{eq:jump_rates}
\end{subequations}
In the limit $\diff t \to 0$, the events involving multiple jumps in the time interval $(t, t + \diff t)$ occur with vanishing probability, so we approximate the system using only the terms in equations\ \eqref{eq:jump_rates} that are linear in $\diff t$. 

\subsection{Global rates}
Now we consider a system of many agents. Again, let $\diff t$ be a small, positive amount of time. Let $\Omega(\numNodesFunctional \to \numNodesFunctional + 1) \diff t$ be the chance that a system with $\numNodesFunctional$ functional agents at time $t$ has $\numNodesFunctional+1$ functional agents at time $t+\diff t$. To first-order in $\diff t$, we have 
\begin{subequations}
\begin{align}
\Omega(\numNodesFunctional \to \numNodesFunctional+1) \diff t = (\numNodes-\numNodesFunctional) \omega(\dysfunctionalState \to \functionalState) \diff t + \mathcal{O}(\diff t^2)
\end{align}
because each of the $\numNodes-\numNodesFunctional$ dysfunctional agents at time $t$ becomes functional at time $t+\diff t$ with probability $\omega(\dysfunctionalState \to \functionalState) \diff t$. Here we are neglecting the probabilities of events that are of order $\mathcal{O}(\diff t^2)$, such as the event in which two agents change from $\dysfunctionalState \to \functionalState$ and one changes from $\functionalState \to \dysfunctionalState$ during the time interval $(t, t+\diff t)$. Similarly, the chance that a system with $\numNodesFunctional$ functional agents at time $t$ has $n-1$ functional agents at time $t+\diff t$ is 
\begin{align}
\Omega(\numNodesFunctional \to \numNodesFunctional-1) \diff t = \numNodesFunctional \omega(\functionalState \to \dysfunctionalState) \diff t + \mathcal{O}(\diff t^2).
\end{align}
\label{eq:global_rates}
\end{subequations}

\subsection{Mean-field approximation of the master equation}
From~\cite[equation 8.95]{MasterEquations}, the expected number $\E \numNodesFunctional(t)$ of functional agents at time $t$ changes over time according to
\begin{align}
\frac{\diff \E \numNodesFunctional(t)}{\diff t} = - \sum_{\ell} \ell \E [\Omega(\numNodesFunctional(t) \to \numNodesFunctional(t) - \ell)]
\label{eq:8p95}
\end{align}
where $\Omega[\numNodesFunctional(t) \to \numNodesFunctional(t) - \ell]$ is the instantaneous rate at which the system jumps from $\numNodesFunctional(t)$ functional agents to $\numNodesFunctional(t) - \ell$ functional agents. 
In our case, by combining\ \eqref{eq:jump_rates} and\ \eqref{eq:global_rates} in\ \eqref{eq:8p95} and dividing both sides of the equation by $\numNodes$, we have
\begin{align*}
\frac{\diff \E \actualFracFunctional(t)}{\diff t} = - \E \left [ \actualFracFunctional(t)  (1-\chanceSuccessSymbol + \decay) \right ] + \E \left [ (1-\actualFracFunctional(t)) \chanceSuccessSymbol \right ] + \mathcal{O}(\diff t^2).
\end{align*}
The \emph{mean-field approximation} \cite[page 255]{MasterEquations} is that fluctuations of $\actualFracFunctional(t) \equiv \numNodesFunctional(t) / \numNodes$ can be ignored because the number of nodes $\numNodes$ is large, so $\E \left [ \actualFracFunctional(t)^k \right ] \approx \left [ \E   \actualFracFunctional(t) \right ]^k \equiv \left [ \fracFunctional(t) \right ]^k$. Thus, because $\chanceSuccessSymbol$ is a polynomial in $\actualFracFunctional$, we approximate $\E \chanceSuccessSymbol \approx \prob{\text{Binomial}\!\left (\degree, \fracFunctional(t) \right ) \geq \threshold}$, which (for simplicity) is how we defined the chance of success $\chanceSuccessSymbol$ in the paper. Neglecting higher-order terms $\mathcal{O}(\diff t^2)$ gives equation\ \eqref{eq:ODE_definition_subequations} in the main text:
\begin{align*}
\frac{\diff \fracFunctional(t)}{\diff t} &= -\fracFunctional (1 - \chanceSuccessSymbol +  \decay) + (1 - \fracFunctional) \chanceSuccessSymbol \\ 
&= \chanceSuccessSymbol - \fracFunctional (1 + \decay).
\end{align*}

\new{
\section{
\newTwo{Deriving the production function and interpreting the production process}}
\label{sec:production_function}
\newTwo{This section has two goals. The first goal is} to derive the production function used in equation\ \eqref{eq:utility} from a standard production function capturing the benefits of specialization \ \cite{romer1987growth,Romer1990,Ethier1982} and a modification of the threshold rule in the O-Ring model\ \cite{Kremer1993}. \newTwo{The second goal is to interpret this production process in a few ways.}

To begin, we fix notation. Consider an agent who attempts to produce a good. Suppose that agent attempts to procure $\degree$ inputs from $\degree$ suppliers. The agent demands quantities of those inputs that we denote by $\demanded_1, \demanded_2, \dots, \demanded_\degree$. 

When arranging for these inputs, it is not clear whether they will be successfully delivered and arrive functional, not broken, not spoiled, and so on. Denote by $x_1, x_2, \dots, x_\degree$ the quantities of inputs that were \newTwo{successfully} delivered in a functional \newTwo{state}: assign $x_i$ a positive quantity if the $i$th input is delivered in functional, working condition, and assign $x_i$ the value $0$ if it is not delivered successfully or if the input is not functional for some reason. The output of the agent is denoted by the production function $f(x_1, x_2, ..., x_\degree).$

\subsection{O-Ring-like threshold rule}

First, we assume that some of the inputs are essential for production to succeed, and without those critical inputs production fails. For simplicity, we do not keep track of which inputs are essential but merely count the number of them. 
Let
\begin{align}
\numSuccesses := \sum_{i \, : \, x_i > 0} 1 \label{eq:num_successful}
\end{align}
denote the number of inputs that are successfully delivered. 
\begin{assumption}[O-Ring-like threshold rule]
Production succeeds [i.e., output $f(x_1, x_2, \dots, x_\degree)$ is positive] if and only if the number of functional inputs $S$ equals or exceeds a threshold number $\threshold$.
\label{assumption:threshold}
\end{assumption}
Assumption\ \ref{assumption:threshold} is related to the production function used in the ``{O-Ring}'' model of economic development\ \cite{Kremer1993}. In the O-Ring model, a product is produced in a sequence of steps. Failure of any of these steps drastically reduces the value of the product.  Here, we generalize the O-Ring rule such that firms can continue to produce as long as redundancy ($\degree - \threshold$) offsets the failures of inputs.

For simplicity, we assume that extra functional inputs are not used:
\begin{assumption}[\newTwo{Number of inputs used}] If more than $\threshold$ functional inputs are available, then only $\threshold$  inputs are used in production.
\label{assumption:substitutes}
\end{assumption}

\newTwo{
Next we interpret this production process in a few ways. 

\subsection{Four interpretations of the production process}
\label{sec:interpretations_production_process}

In this subsection, we briefly reiterate the primary interpretation of the production process that we use in the main text, and then we give three alternative interpretations.

\paragraph{Firms simultaneously procure inputs, and we abstract from which pairs are substitutable}
\label{sec:simultaneous_procure}

The interpretation that we focus on in the manuscript is that a firm procures $\degree$ inputs simultaneously from other agents, and some of those $\degree$ inputs are perfect substitutes for one another. In reality, pairs of inputs have differing amounts of substitutability, and considering this complexity in input--output networks is an ongoing challenge for theorists\ \cite{Carvalho2014_Perspective}. In this interpretation, we simplify by ignoring which specific pairs are substitutes. 

\paragraph{Firms sequentially procure inputs}
\label{sec:sequential_interpretation}

In this interpretation, we do track which inputs are substitutable. Consider firms searching for inputs over a certain period of time. If the firm does not successfully produce within that period of time, then the firm fails to produce anything and becomes dysfunctional, for example because the inputs spoil, or because debts must be repaid, or because the growing season has passed. 

Suppose that each of the $\threshold$ different inputs needed (from Assumption\ \ref{assumption:threshold}) is produced by a large set of identical firms. Each of these sets of firms (or ``industries'') produces a distinct good. Other than that, these industries are completely symmetric, so the (expected) fraction of functional firms in each industry is the same value, $\fracFunctional(t)$. 

A firm commits in advance to a budget for how many times ($\degree$) it can attempt to procure a functional input from one of these $\threshold$ many industries. Each attempt to procure a functional input is an independent Bernoulli trial with probability $\fracFunctional(t)$. The firm procures inputs until either it gets $\threshold$ or more functional inputs (with one functional input from each of the $\threshold$ industries) or until it has exhausted its $\degree$ attempts. 

Although the firm draws Bernoulli trials sequentially, suppose for theoretical purposes that we knew the results of the next $\degree$ Bernoulli trials. The firm will observe some or all of these trials, depending on their outcomes. The chance that the firm succeeds in producing is the chance that these $\degree$ trials contain at least $\threshold$ successes. In other words, the chance that the firm succeeds in producing is the chance that a binomial random variable with parameters $\degree$ and $\fracFunctional(t)$ is greater than or equal to $\threshold$. Provided that the time period in which the firm can attempt to procure $\degree$ inputs is brief enough so that $\fracFunctional(t)$ can be assumed to be constant, this interpretation is consistent with the model studied in this paper.

\paragraph{Search intensity}
\label{sec:search_intensity}

Another interpretation involves search intensity and research costs. We illustrate this interpretation using the case study on drip irrigation\ \cite{Garb2014} summarized in the manuscript. 

Consider a farmer preparing for the upcoming growing season. The farmer decides to attempt to use drip irrigation, which requires a large number ($\threshold$) of inputs in order to succeed, such as the $6$-tuple $(\text{seeds, labor, rubber tubes, pump, know how, municipal water})$.

The farmer has a limited amount of time to research which kinds of seeds, pump, and rubber tubes best match the local conditions, including the acidity of the soil, the pressure of the water supply, which animals might chew the tubes, and so on. The farmer attempts to procure each of the $\threshold = 6$ inputs and can reduce the risk that the crop fails during the coming growing season by spending more time and effort researching which inputs are most compatible with the local conditions. 

The mean-field approximation studied here captures this scenario in a stylized way. 
We denote by $x \in \{0, 1, 2, \dots\}$ a discrete measurement of the intensity with which the farmer searches for the best or most compatible inputs. 
We assume that the probability that the farmer successfully produces is a function of the complexity $\threshold$ of the technology, the intensity $x$ of her search, and the reliability $\fracFunctional(t)$ of the economy. Formally, in our model this probability is identified with the probability that a binomial random variable with parameters $\threshold + x$ and $\fracFunctional(t)$ is at least $\threshold$. The ``research costs'' are $\alpha \times (\threshold + x)$. A more natural model might have the search intensity reduce the failure rate in a continuous way, but this formulation with a binomial random variable has qualitatively similar behavior.

\paragraph{Production flexibility}
\label{sec:production_flexibility}
The fourth and final interpretation considers flexibility in what a firm can produce from a certain collection of inputs. Assume that a firm tries to procure $\degree$ distinct inputs that can be combined to produce $\binom{\degree}{\threshold}$ many distinct products that involve $\threshold$ inputs each. 
All of these products are assumed to have similar market potential. When the number of available functional inputs is $\threshold$ or greater, a firm is indifferent to which output to produce, so it randomizes among them.

The difference $\degree - \threshold$ is a stylization of manufacturing flexibility studied in management literature\ \cite{DSouza2000toward}, in particular of ``(product) range flexibility'', which is the number of products a firm can produce without prohibitive switching costs\ \cite{DSouza2000toward,Upton1997process}. 
Greater flexibility enables the firm to continue producing even if some inputs are not available in functional form. 
}

\newThree{
This interpretation is also conceptually related to the microfoundation of a ``global production function'' developed in \cite{jones2005shape}. 
In\ \cite{jones2005shape}, the global production function computes the maximum amount of output per worker that can be produced using an available set of ``production techniques'', which are mappings from inputs (such as capital and labor) to output. 
To see the conceptual relationship with the ``production flexibility'' interpretation of our model, note that we can interpret number $\degree \choose \threshold$ as the number of production techniques that a firm can use to produce output, in the spirit of \cite{jones2005shape}. 
Each production technique is a mapping from $\threshold$ many inputs, $(x_1, \dots, x_\threshold)$, to output. 
The output of a production technique is given by equation\ \eqref{eq:ProdFunc} (found below) provided that enough inputs are functional; otherwise, the output is 0. 
The production techniques available to a firm are endogenously determined by the firm's choice of $(\degree, \threshold)$ (which are chosen based on the state of the economy $\fracFunctional(t)$). 
A firm optimally chooses the production technique that results in the highest output. 
Given the complexity of the problem, deriving the shape of a firm's global production function, as defined in \cite{jones2005shape}, is outside the scope of this paper and is left for the future research. 
}

\subsection{Positive returns due to specialization}

Next we introduce a standard production function that has been widely used in the economic literature to model the gains from specialization, interpreting such specialization as an increase in the number of distinct inputs combined to carry out production. 
\begin{assumption}Suppose an agent has enough functional inputs ($\numSuccesses \geq \threshold$). Rearrange the indices of the inputs $x_1, x_2, \dots, x_\degree$ so that $x_1, x_2, \dots, x_\threshold$ are all positive. An agent's output is given by the following CES (constant elasticity of substitution) production function:
 \begin{align}\label{eq:ProdFunc}
f(x_1, x_2, ..., x_\threshold) = \left ( \sum_{i=1}^\threshold \left ( x_i \right )^\rho \right )^{\frac{1}{\rho}},
\end{align}
where $\rho \in (1/2, 1)$.
\label{assumption:CES}
\end{assumption}

To understand why\ \eqref{eq:ProdFunc} indeed embodies gains from specialization, consider the following thought experiment. Suppose input prices are fixed -- for simplicity, let all of them be equal -- and consider a firm that has at least $\threshold$ many functional inputs available (i.e., $\numSuccesses \geq \threshold$). Rearrange the indices of the inputs $x_1, x_2, \dots, x_\degree$ so that $x_1, x_2, \dots, x_\threshold$ are all positive. If the firm wants to maximize production and has a certain amount of funds $E$ available to buy inputs, then it must demand the same amount of each of the intermediate inputs used. Thus, by Assumption  \ \ref{assumption:substitutes}, it can use for example the first $ \threshold$ and demand $u_i = x_i = E / \threshold$ for $i = 1, 2, \dots, \threshold$, so the output is
\begin{align}
f(x_1, x_2, ..., x_\threshold) 
= E \times \threshold^{(1/\rho) - 1} = E \threshold^\beta
\label{eq:spread_E_evently}
\end{align}
where $\beta := (1 / \rho) - 1$. Then, if we further make the usual assumption that the elasticity of substitution $\sigma := 1 / (1 - \rho)$ is positive, we arrive at the conclusion that $\rho \in (1/2, 1)$, so the elasticity of substitution $1 / (1 - \rho)$ is positive and $\beta \in (0, 1)$. In this case, we observe from equation\ \eqref{eq:spread_E_evently} that, as the number $\threshold$ of intermediate inputs purchased with the same funds $E $ increases, so does the induced output. Hence, in an natural way, we may describe the production technology as displaying \emph{positive returns due to specialization}\ \cite{romer1987growth,Romer1990,Ethier1982}; that is, efficiency grows as the production process is subdivided into a larger set of more ``specialized tasks." 

Without any essential loss of generality, we may normalize the expenditure $E$ to one. Then the production function becomes
\begin{align}\label{eq:ProdFunc_simplified}
f(x_1, x_2, ..., x_m) = 
\begin{cases} 
{\tau}^{\beta} & \text{if} \quad S \equiv \sum_{i \, : \, x_i > 0} 1 \geq \threshold \\
0 & \text{otherwise}
\end{cases}.
\end{align}
Equation\ \ref{eq:ProdFunc_simplified} is the production function that we use in equation\ \eqref{eq:utility} of the main text to determine the utility attained at any level of specialization reflected by $\threshold$, given the prevailing fraction $\fracFunctional$ of functional agents and the cost induced by the number $m$ of inputs they attempt to procure.}

\newThree{
\section{Interpretations of the marginal cost $\alpha$ of finding and procuring inputs\label{sec:interpret_alpha}}
We assume that attempting to get an input from another agent costs $\alpha > 0$. This parameter can have one of several interpretations:
\begin{description}
\item [Search costs] The parameter $\alpha$ may capture the time and energy required to find each new supplier.
\item [Maintaining multiple suppliers for an input] Maintaining multiple suppliers for the same input can be costly when that requires changing the product, working with suppliers to develop alternatives, or overcoming quality issues with alternative inputs\ \cite{ProvitiReport}.
\item [Incentivising suppliers to have multiple manufacturing sites] If a firm has a certain crucial input with just one supplier (a ``strategic component''), it may ``provid[e] incentives to [those] suppliers to have multiple manufacturing sites in different regions''\ \cite{SupplyChainDisruptionsHBR2014}.
\item [Maintaining inventory of inputs] Inventory costs can be high: ``as product life cycle shortens and as product variety increases, the inventory holding and obsolescence costs of these additional safety stock inventories could be exorbitant''\ \cite{Tang2006}.
\end{description}
}

\section{Strategies that could be a best response}
\label{sec:strategies_that_could_be_best_response}

Recall that the agents' decision problem is to maximise the utility\ \eqref{eq:utility} 
$$\utility{\degree}{\threshold}{\fracFunctional(t)}{\alpha}{\beta} 
= \prob{\text{Binomial}({\degree, \fracFunctional}) \geq \threshold} \threshold^\beta - \alpha \degree$$ over all pairs of non-negative integers $(\degree, \threshold) \in \mathbb{N}^2$, where $\mathbb{N} \equiv \{0, 1, 2, 3, \dots\}$.
The set of strategies $(\degree, \threshold)$ that could be a best response turns out to be a finite set, which 
enables numerical simulations. 
\begin{lemma}
A best response $(\degree^*, \threshold^*)$ must belong to the set 
\begin{align}
\{(\degree, \threshold) \in \mathbb{N}^2 : \degree = 0 \text{ \emph{or} } 0 < \threshold \leq \degree < \threshold^{\beta} / \alpha \}.  
\end{align}
\label{lem:could_be_best_response}
\end{lemma}
\begin{proof}
First observe that if $\degree = 0$ then the utility is zero. If $0 < \degree < \threshold$, then the utility is negative, so a best response cannot have $0 < \degree < \threshold$. Hence a best response must have $\degree = 0$ or $\degree \geq \threshold > 0$. 

In the latter case (with $\degree \geq \threshold > 0$), the utility must exceed zero (the utility obtained with $\degree = 0$), so 
\begin{align*}
\prob{\text{\emph{Binomial}}({\degree, \fracFunctional}) \geq \threshold} \threshold^\beta - \alpha \degree > 0,
\end{align*}
or, after rearranging and dividing by $\alpha >0$, 
\begin{align}
\degree < \prob{\text{\emph{Binomial}}({\degree, \fracFunctional}) \geq \threshold} \frac{\threshold^\beta}{\alpha} \leq \frac{\threshold^\beta}{\alpha}. 
\label{eq:upper_bound_m}
\end{align}
$\hfill \blacksquare$
\end{proof}

\begin{corollary}
A best response $(\degree^*, \threshold^*)$ with $\degree^* > 0$ must satisfy $\threshold^* \leq \alpha^{{-1}/(1-\beta)}$ and $\degree^* < \alpha^{-1/(1-\beta)}$. 
\end{corollary}
\begin{proof}
From Lemma\ \ref{lem:could_be_best_response}, we know that if $(\degree, \threshold)$ is a best response with $\degree > 0$ then $\threshold \leq \degree < \threshold^{\beta} / \alpha$. 
Now we equate these lower and upper bounds on $\degree$. 
The equation $\threshold = \threshold^{\beta} / \alpha$ has a two solutions: $\threshold = 0 $ and $\threshold = \alpha^{{-1}/(1-\beta)}$. The latter provides provides an upper bound on $\threshold$ in the best response. 
From Lemma~\ref{lem:could_be_best_response} we know that a best response must have $\degree < \threshold^\beta / \alpha$, so  $\degree < \threshold^\beta / \alpha \leq \alpha^{-1/(1-\beta)}$. 
$\hfill \blacksquare$
\end{proof}

Numerical calculations of a best response can be sped up slightly using the following two observations:
\begin{enumerate}
\item If $\degree = \threshold > 0$ is a best response, then $\utility{\degree}{\degree}{\fracFunctional(t)}{\alpha}{\beta}  = \fracFunctional(t)^\degree \degree^\beta - \alpha m >  \utility{0}{0}{\fracFunctional(t)}{\alpha}{\beta}  = 0$, which can be rewritten as 
\begin{align}
m<\frac{(\beta -1)}{\log \left [ \fracFunctional(t) \right ]} W\left(\frac{\left(\frac{1}{\alpha}\right)^{\frac{1}{1-\beta}} \log \left [ \fracFunctional(t) \right ]}{\beta -1}\right),
\label{eq:solution_max_m_for_m_equals_tau}
\end{align}
where $W$ is the product logarithm (i.e., the Lambert W function). This upper bound on the diagonal $\degree = \threshold$ can be tighter than the one in inequality~\eqref{eq:upper_bound_m}. 
\item If $\fracFunctional(t) = 1$, then every strategy with $\degree \geq \threshold$ will certainly succeed [i.e., $\prob{\text{Binomial}({\degree, \fracFunctional(t)}) \geq \threshold} = 1$], so the best response must have $\degree = \threshold$, and the first-order condition for the utility $\degree^\beta - \alpha \degree$ shows that the best-response $\degree = \threshold$ is either the floor or the ceiling of $\left({\beta/\alpha}\right)^{{1}/{1-\beta}}$ (whichever one results in more utility).
\end{enumerate}

If practice, if the best response is $(0, \threshold)$, where $\threshold \in \mathbb{N}$, then we take the best response to be $(0,0)$ because it is not economically meaningful to have $\threshold > \degree = 0$ 
(it would mean that an agent requires some inputs to produce but does not attempt to acquire any inputs). 

In simulations, we compute the set of strategies that could be best responses as follows:
\begin{align}
\begin{cases}
\{(0,0)\} & \text{if } F= 0 \\
\{
(\lfloor \gamma \rfloor, \lfloor \gamma \rfloor), 
(\lceil \gamma \rceil, \lceil \gamma \rceil)\} & \text{if } F=1 \\
\{(\degree, \threshold) \in \mathbb{N}^2 : 
0 \leq \degree \leq \alpha^{-1/{(1-\beta)}}, \text{and either } \\
\degree=\threshold \text{ satisfies\ \eqref{eq:solution_max_m_for_m_equals_tau} or } 
0 < \threshold < \degree < \threshold^\beta / \alpha
\}
& \text{if } 0<F<1
\end{cases}
\label{eq:compute_strategies_that_could_be_best_response}
\end{align}
where 
\begin{align*}
\gamma \equiv \left (\frac{\beta}{\alpha} \right )^{\frac{1}{1-\beta}}.
\end{align*}
Given $(\alpha, \beta, \fracFunctional(t))$, we compute the finite set of strategies that could be a best response [given by expression~\eqref{eq:compute_strategies_that_could_be_best_response}] and among those strategies we select the strategy with the highest utility [given by equation\ \eqref{eq:utility} in the main text].

\section{Downturns in rich, highly functional\\ economies}
\label{appendix:downturns}

Here we show that rich, highly functional economies undergo occasional spikes in dysfunction, akin to economic recessions. \new{These spikes in disruption occur even if there are no exogenous failures (i.e., $\decay = 0$): because agents can economize too much on redundancy, failures can spread.}

\begin{lemma}[Without exogenous failures, a completely functional economy is in steady state]
\new{For $\decay = 0$,} there is a steady state at $\fracFunctional(t) = 1$, at which agents best respond by choosing $\degree^* = \threshold^*$. Furthermore, $\degree^* = \threshold^*$ equals the floor or ceiling of $(\beta / \alpha) ^ {1/(1-\beta)}$ (whichever yields more utility). 
\label{lem:steady_state_F=1}
\end{lemma}
\begin{proof}
In a completely functional economy [with $\fracFunctional(t) = 1$] without exogenous failures \new{($\decay = 0$)}, the ODE for $\fracFunctional(t)$ [from equation\ \eqref{eq:ODE_definition_subequations} in the main text] is 
\begin{align}
\diff \fracFunctional / \diff t = \chanceSuccessSymbol - 1 
\label{eq:ODEatF=1}
\end{align}
for $\threshold > 0$ (and $\diff \fracFunctional / \diff t := 0$ for $\threshold = 0$). 
Because $\fracFunctional = 1$, agents certainly succeed in producing (i.e., $\chanceSuccessSymbol = 1$) as long as $\degree \geq \threshold > 0$. 
If $\chanceSuccessSymbol = 1$, then choosing $\degree$ to be greater than $\threshold$ presents only costs and no benefits, so agents choose $\degree^* = \threshold^*$. 
By computing the value of $\threshold$ that maximises the utility function $\threshold^\beta - \alpha \threshold$, we see that the best response at $\fracFunctional(t) = 1$ has $\degree^* = \threshold^*$ equal to either the floor or ceiling of $(\beta / \alpha) ^ {1/(1-\beta)}$ (whichever yields more utility). 

In summary, if $\degree^* = \threshold^* = 0$ then $\diff \fracFunctional / \diff t = 0$ [by definition; see the text after equation\ \eqref{eq:ODE_definition_subequations} in the main text], and if $\degree^* = \threshold^* > 0$ then $\chanceSuccessSymbol = 1$ and therefore $\diff \fracFunctional / \diff t = 0$ from equation\ \eqref{eq:ODEatF=1}. $\hfill \blacksquare$
\end{proof}

Next we show in Lemma\ \ref{lem:F=1noredundancy} that the steady state at $\fracFunctional(t) = 1$ is not stable to perturbations in $\fracFunctional(t)$ \new{because agents do not have any redundancy (i.e., $\degree^* - \threshold^* = 0$) for $\fracFunctional$ just below $1$.}

\begin{lemma}[The best response has no redundancy at and just below $\fracFunctional(t) = 1$]
Assume that $\alpha > 0$, $\beta \in (0,1)$, and $\decay \geq 0$.  
There exists $\bar \fracFunctional \in (0,1)$ such that the best response $(\degree^*, \threshold^*)$ satisfies $\degree^* = \threshold^*$ for $\fracFunctional(t) \in (\bar \fracFunctional , 1]$. 
\label{lem:F=1noredundancy}
\end{lemma}
\begin{proof} 
We will prove the claim using the continuity of the utility function $\utility{\degree}{\threshold}{\fracFunctional(t)}{\alpha}{\beta}
$ [defined in equation\ \eqref{eq:utility} in the main text] as a function of $\fracFunctional(t)$. $\utilitySymbol$ is continuous in $\fracFunctional(t)$ because it is a polynomial in $\fracFunctional(t)$ of degree $\degree$. 

From Lemma\ \ref{lem:steady_state_F=1}, we know that for $\fracFunctional(t) = 1$ the best response is $\degree^* = \threshold^*$ and that $\threshold^*$ equals the floor or the ceiling of $(\beta / \alpha) ^ {1/(1-\beta)}$ (whichever yields more utility). We denote this value of $\threshold^*$ by $\phi$. 

From Lemma\ \ref{lem:could_be_best_response}, 
we know that the set of strategies that could be a best response is finite and depends on $\alpha$ and on $\beta$ but not on $\fracFunctional$. 
Denote using $\Phi$ the set of strategies that could be a best response, which is given in Lemma\ \ref{lem:could_be_best_response}. 

Let $\delta_k$ denote the difference between the first-best utility and $k$th-best utility, 
where the $k$-th best utility is defined to be the $k$th-highest utility among the set of strategies $\Phi$ that could be best responses. 

Let $\delta := 1/2 \times \min \{\delta_k : 2 \leq k \leq | \Phi | \}$. We know that $\delta > 0$ because it is a minimum of finitely many positive numbers (because $|\Phi| < \infty$). Because $\utilitySymbol$ is continuous in $\fracFunctional(t)$, there exists $\theta \in (0, 1)$ such that if $\fracFunctional(t) > 1 - \theta$ then 
\begin{align}
| \utility{\degree}{\threshold}{\fracFunctional(t)}{\alpha}{\beta}
  -\utility{\degree}{\threshold}{1}{\alpha}{\beta}
 | < \delta 
\end{align}
for all $(\degree, \threshold) \in \Phi$. 

Now fix $\fracFunctional_0 \in (1-\theta, 1)$ and let $(\degree_0, \threshold_0) \in \Phi$ be one of the possible best responses different from $(\phi, \phi)$ [which is the best response at $\fracFunctional(t) = 1$]. Let $k$ denote the utility-rank of this strategy $(\degree_0, \threshold_0)$ at $\fracFunctional(t) = 1$, meaning that the strategy  $(\degree_0, \threshold_0)$ has the $k$-th highest utility at $\fracFunctional(t) = 1$. Then 
\begin{subequations}
\begin{align}
\utility{\phi}{\phi}{\fracFunctional_0}{\alpha}{\beta}
&> 
\utility{\phi}{\phi}{1}{\alpha}{\beta}
- \delta & \text{by continuity of $\utilitySymbol \left [ \dots, \fracFunctional(t), \dots \right ]$} \\
& \geq 
\utility{\phi}{\phi}{1}{\alpha}{\beta}
- \delta_k / 2 & \text{by definition of $\delta$} \\
& = 
\utility{\degree_0}{\threshold_0}{1}{\alpha}{\beta}
+  \delta_k / 2  & \text{by definition of $\delta_k$}  \\
& \geq 
\utility{\degree_0}{\threshold_0}{1}{\alpha}{\beta}
+ \delta & \text{by definition of $\delta$}  \\
& > 
\utility{\degree_0}{\threshold_0}{\fracFunctional_0}{\alpha}{\beta}
 & \text{by continuity of $\utilitySymbol \left [ \dots, \fracFunctional(t), \dots \right ]$},
\end{align}
\end{subequations}
so $(\phi, \phi)$ is the best response for all $\fracFunctional \in (1-\theta, 1]$. This argument completes the proof with $\bar \fracFunctional := 1-\theta \in (0,1)$. 
$\hfill \blacksquare$
\end{proof}

Finally, we arrive at our main result of this appendix: if the marginal cost of each attempted input ($\alpha$) is small enough, then a sufficiently functional economy undergoes spikes in disruptions, akin to recessions. These spikes in disruption occur even in the absence of exogenous failures ($\decay = 0$).

\begin{theorem}[Highly functional economies undergo spikes in disruption]
If $\alpha < 2^\beta - 1$, $\beta \in (0, 1)$, and $\decay \geq 0$, then there exists $\widetilde \fracFunctional \in (0,1)$ such that $\diff F / \diff t < 0$ for $\fracFunctional \in (\widetilde \fracFunctional, 1)$. 
\label{thm:downturn}
\end{theorem}
\begin{proof}
From Lemma~\ref{lem:F=1noredundancy} we know that there exists $\bar \fracFunctional \in (0,1)$ such that the best response $(\degree^*, \threshold^*)$ satisfies $\degree^* = \threshold^*$ for $\fracFunctional \in (\bar \fracFunctional,1]$, 
and from Lemma~\ref{lem:steady_state_F=1} we know that this 
$\degree^* = \threshold^*$ equals the floor or ceiling of $(\beta / \alpha) ^ {1/(1-\beta)}$ (whichever yields more utility), which we denote by $\phi$. Note that $\phi$ is nonincreasing in $\alpha$. 

\new{Now we determine how small $\alpha$ needs to be in order for the best response to satisfy $\degree^* = \threshold^* \geq 2$. To do so,} we equate utilities from pairs of strategies:
\begin{itemize}
\item An agent is indifferent between the strategies $(\degree, \threshold) = (1,1)$ and $(0,0)$ when 
$\fracFunctional(t) - \alpha = 0$. 
This observation, together with the observation that $\phi$ is nonincreasing in $\alpha$, implies that for $\fracFunctional(t) = 1$ and $\alpha \geq 1$ the best response is $(0,0)$. 
\item An agent is indifferent between the strategies $(\degree, \threshold) = (1,1)$ and $(2,2)$ when $\fracFunctional^2 2^\beta - 2 \alpha = \fracFunctional(t)-\alpha$, which is a quadratic polynomial in $\fracFunctional(t)$. The positive root is $\fracFunctional_+ := 2^{-(\beta +1)} \left(1 - \sqrt{1 - \alpha  2^{\beta +2}}\right)$. This positive root $\fracFunctional_+$ is equal to $1$ when $\alpha = 2^\beta - 1$, which is smaller than $1$ for any $\beta \in (0, 1)$. Therefore, for $\alpha \in (2^\beta - 1, 1)$ and $\fracFunctional(t) = 1$ the best response is $(\degree^*, \threshold^*) = (1,1)$. Moreover, for $\alpha \in (0, 2^\beta - 1)$ and $\fracFunctional(t) = 1$ the best response $(\degree^*, \threshold^*)$  satisfies $\degree^* = \threshold^* \geq 2$.
\end{itemize}

Finally, note that if $\degree = \threshold \geq 2$, 
then the chance of successfully producing is $\chanceSuccess{\degree}{\threshold}{\fracFunctional(t)} 
= \fracFunctional(t)^{\threshold}$, which 
is less than $\fracFunctional(t)$ for $\fracFunctional(t) \in (0,1)$ because $\threshold \geq 2$. 
\new{In this case, the master equation\ \eqref{eq:ODE_definition}
\begin{align}
\diff \fracFunctional / \diff t = \fracFunctional(t)^{\threshold} - \fracFunctional(t) (1 + \decay)
\end{align}
is negative for all  $0 < \fracFunctional(t) < 1$ and any $\decay \geq 0$. Thus, the claim holds for $\widetilde \fracFunctional := \bar \fracFunctional$.  $\hfill \blacksquare$}
\end{proof}

\new{
\section{Modified model in which agents tend to keep functional suppliers and tend to choose popular suppliers\label{sec:modified_model}}

Here we modify the choice of suppliers in two ways. First, agents tend to keep their suppliers who were functional. Second, they bias their search toward suppliers who already have many customers. (When $i$ requests an input from $j$, we consider $i$ a customer of $j$ until $i$ tries to produce again.) Specifically, when agent $i$ tries to produce, it chooses $\degree$ suppliers from the population with replacement according to the following weights. Agent $j$ with $k$ customers has weight 
\begin{align}
w_j = \stickiness^{\mathcal{S}(j, i)} (1 + k)^{\xi},
\label{eq:weight}
\end{align}
where $\mathcal{S}(j, i)$ is $1$ if $j$ successfully delivered an input to $i$ last time $i$ tried to produce, and otherwise $\mathcal{S}(i, j) = 0$. 
The parameter $\stickiness$ captures how ``sticky'' customer--supplier relationships are, while the ``preferential attachment'' parameter $\xi$ determines how strongly agents tend to choose suppliers who have high out-degree $k$.\footnote{\new{$\xi$ is the parameter of nonlinear preferential attachment\ \cite{Krapivsky2000connectivity}; in that model, $\xi=1$ generates a power-law out-degree distribution. Reference \cite{Fujiwara2010} found a power law out-degree distribution in Japan's national, firm-level input--output network.}} 
Our original model is recovered with $(r, \xi) := (1, 0)$.

Assuming $r > 1$ or $\xi > 0$ generates the same qualitative results insights of the model (the poverty trap, overshooting in complexity, and rise-and-fall of buffers). However, it does generate two interesting effects that we illustrate using numerical simulations. First, with $\stickiness > 1$ and $\xi = 1$, the economy is less likely to fall into the trap. Because the model is stochastic and the simulations have finitely many agents, the economy can fall into the poverty trap (and remain there) even if the master equation\ \eqref{eq:ODE_definition_subequations} predicts growth ($\diff \fracFunctional / \diff t > 0$). With $\stickiness > 1$, agents tend to keep their functional suppliers, so they are more likely to successfully produce. Consequently, functional agents tend to accumulate customers. If we also have preferential attachment ($\xi > 0$), then agents further bias their search for suppliers toward functional agents because having many customers is correlated with being functional.

This biased search for  for functional suppliers has a downside, however. Once the economy becomes complex and highly functional, it can rely on a handful of agents who supply almost all others. These ``hubs'' can become dysfunctional because they need inputs from dysfunctional suppliers or because they are hit by one of the exogenous failures that occur at rate $\decay$. When that happens, the economy is so brittle and so dependent on those hubs that it suffers cascading disruptions and a drastic downturn, as illustrated by the large drops in the time-series $\fracFunctional(t)$ in Supplementary Figure\ \ref{fig:sticky_pa}(b).
}


\begin{figure}[htb]
\begin{center}
\includegraphics{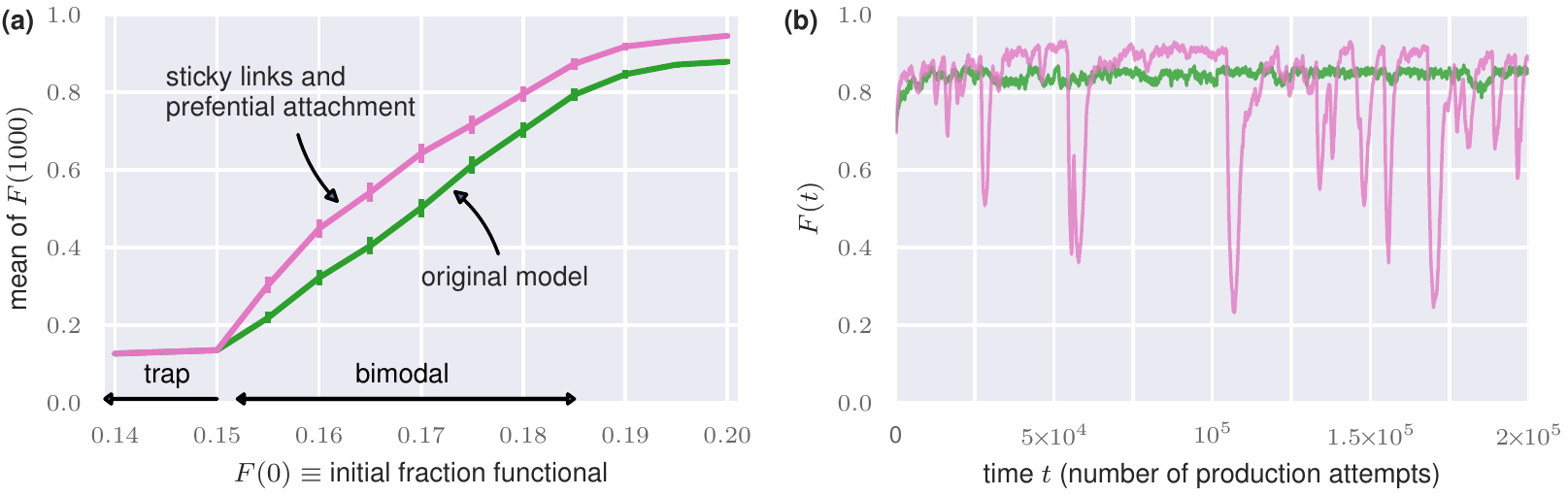}
\caption{
\new{
A tendency to keep functional suppliers ($r>1$) and to choose popular suppliers ($\xi = 1$) makes it (a) easier to escape poverty but (b) once complex and developed, the economy is more fragile.
In both panels, the green curve is the original model [$(r, \xi) = (1, 0)$]; the magenta curve is the ``modified model'' with sticky links ($r=2000$) and linear preferential attachment ($\xi = 1$); and $(\alpha, \beta, \epsilon) = (0.15, 0.4, 10^{-4})$. The economy has $N=200$ agents in (a) and $1000$ in (b).
Panel (a) shows the mean fraction of functional agents in the long-run $\fracFunctional(10^3)$ as a function of the initial condition $\fracFunctional(0)$; error bars denote $2 \times \text{s.e.m.}$ for $10^3$ simulations; 
the mean $\fracFunctional(10^3)$ is significantly larger in the modified model for all $F(0) \in [0.155, 0.2]$ ($p$-value $< 10^{-5}$, two-sided Mann-Whitney $U$ test). 
Panel (b) shows two representative time-series $\fracFunctional(t)$ with $\fracFunctional(0) = 0.7$. 
The standard deviation of $F(t)$ is significantly larger in the modified model [on average, $8.6$ times larger; $p$-value $= 10^{-251}$ in a two-sided \emph{t}-test with $200$ replicas of the simulations in (b)].
}}
\label{fig:sticky_pa}
\end{center}
\end{figure}


\end{document}